\documentclass[prodmode,acmtalg]{acmsmall}
\usepackage{natbib}
\usepackage{enumerate}

\renewcommand{\log}{\lg}
\newcommand{\Oh}[1]
    {\ensuremath{\mathcal{O}\!\left( {#1} \right)}}
\newcommand{\rank}
    {\ensuremath{\mathrm{rank}}}
\newcommand{\select}
    {\ensuremath{\mathrm{select}}}
\newcommand{\occ}
    {\ensuremath{\mathrm{occ}}}
\newcommand{\ignore}[1]{}
    
\begin{document}

\markboth{D. Belazzougui et al.}{Range Majorities and Minorities in Arrays}

\title{Range Majorities and Minorities in Arrays}
\author{DJAMAL BELAZZOUGUI \affil{CERIST, Algeria}
    TRAVIS GAGIE \affil{University of Helsinki, Finland}
    J.\ IAN MUNRO \affil{University of Waterloo, Canada}
    GONZALO NAVARRO \affil{University of Chile}
    YAKOV NEKRICH \affil{University of Waterloo, Canada}}

\begin{abstract}
Karpinski and Nekrich (2008) introduced the problem of parameterized range
majority, which asks us to preprocess a string of length $n$ such that, given the
endpoints of a range, one can quickly find all the distinct elements whose
relative frequencies in that range are more than a threshold $\tau$.
Subsequent authors have reduced their time and space bounds such that, when
$\tau$ is fixed at preprocessing time, we need either $\Oh{n \log (1 / \tau)}$
space and optimal $\Oh{1 / \tau}$ query time or linear space and $\Oh{(1 /
\tau) \log \log \sigma}$ query time, where $\sigma$ is the alphabet size.  In
this paper we give the first linear-space solution with optimal $\Oh{1 /
\tau}$ query time, even with variable $\tau$ (i.e., specified with the query).
For the case when %$\tau$ is fixed or 
$\sigma$ is polynomial on the computer word size, our space is optimally compressed according to the symbol frequencies in the string. Otherwise, either the compressed space is increased by an arbitrarily small constant factor or the time rises to any function in $(1/\tau)\cdot\omega(1)$. We obtain the same results on the complementary problem of parameterized range minority introduced by Chan et al.\ (2015), who had achieved linear space and $\Oh{1 / \tau}$ query time with variable $\tau$.
\end{abstract}

%ex:\category{E.2}{Data storage representations}{}



%\acmformat{Djamal Belazzougui, Travis Gagie, Ian Munro, Gonzalo Navarro, and Yakov Nekrich, 2016. Succinct Data Structures for Frequency-Sensitive Queries in Ranges.}

\begin{bottomstuff}
Funded in part by 
%European Union's Horizon 2020 research and innovation programme under the Marie Sk{\l}odowska-Curie grant agreement No 690941, by 
Academy of Finland grant 268324 and Millennium Nucleus Information and Coordination in Networks ICM/FIC RC130003 (Chile). An early partial version of this article appeared in {\em Proc. WADS 2013}.

Authors' addresses: D. Belazzougui, Research Center on Technical and Scientific
Information, Algeria; T. Gagie, Department of Computer Science, University of Helsinki; I. Munro and Y. Nekrich, David Cheriton School of Computer Science, University of Waterloo, Canada; G. Navarro, Department of Computer Science, University of Chile.
\end{bottomstuff}

\maketitle

\section{Introduction} \label{sec:intro}

Finding frequent elements in a dataset is a fundamental operation in data mining. The most frequent elements can be difficult to spot when all the elements have nearly equal frequencies.  In some cases, however, we are interested in the most frequent elements only if they really are frequent.  For example, \citet{MG82} showed how, given a string and a threshold \(0 < \tau \leq 1\), by scanning the string twice and using $\Oh{1 / \tau}$ space we can find all the distinct elements whose relative frequencies exceed $\tau$. These elements are called the $\tau$-majorities of the string. If the element universe is $[1..\sigma]$, their algorithm can run in linear time and $\Oh{\sigma}$ space. \citet{DLM02} rediscovered the algorithm and deamortized the cost per element; \citet{KSP03} rediscovered it again, obtaining $\Oh{1/\tau}$ space and linear randomized time.  As \citet{CM03} put it, ``papers on frequent items are a frequent item!''.

\citet{KMS05} introduced the problem of preprocessing the string such that later, given the endpoints of a range, we can quickly return the mode of that range (i.e., the most frequent element).  They gave two solutions, one of which takes $\Oh{n^{2 - 2 \epsilon}}$ space for any fixed positive \(\epsilon \leq 1 / 2\), and answers queries in $\Oh{n^\epsilon \log \log n}$ time; the other takes $\Oh{n^2 \log \log n / \log n}$ space and answers queries in $\Oh{1}$ time. \citet{Pet08} reduced \citeauthor{KMS05}$\!$'s first time bound to $\Oh{n^\epsilon}$ for any fixed non-negative \(\epsilon < 1 / 2\), and \citet{PG09} reduced the second space bound to $\Oh{n^2 \log \log n / \log^2 n}$. \citet{CDLMW14} gave an $\Oh{n}$ space solution that answers queries in $\Oh{\sqrt{n / \log n}}$ time.  They also gave evidence suggesting we cannot easily achieve query time substantially smaller than $\sqrt{n}$ using linear space; however, the best known lower bound 
\cite{GJLT10} says only that we cannot achieve query time \(o ( \log (n) / \log (s w / n) )\) using $s$ words of $w$ bits each.  Because of the difficulty of supporting range mode queries, \citet{BKMT05} and \citet{GJLT10} considered the problem of approximate range mode, for which we are asked to return an element whose frequency is at least a constant fraction of the mode's frequency.

\citet{KN08} took a different direction, analogous to \citeauthor{MG82}' approach, when they introduced the problem of preprocessing the string such that later, given the endpoints of a range, we can quickly return the $\tau$-majorities of that range.  We refer to this problem as parameterized range majority.  Assuming $\tau$ is fixed when we are preprocessing the string, they showed how we can store the string in $\Oh{n (1 / \tau)}$ space and answer queries in $\Oh{(1 / \tau) (\log \log n)^2}$ time.  They also gave bounds for dynamic and higher-dimensional versions.  \citet{DHMNS13} independently posed the same problem and showed how we can store the string in $\Oh{n \log (1 / \tau + 1)}$ space and answer queries in $\Oh{1 / \tau}$ time.  Notice that, because there can be up to \(1 / \tau\) distinct elements to return, this time bound is worst-case optimal. \citet{GHMN11} showed how to store the string in compressed space --- i.e., $\Oh{n (H + 1)}$ bits, where 
$H$ is the entropy of the distribution of elements in the string --- such that we can answer queries in $\Oh{(1 / \tau) \log \log n}$ time. Note that $H \le \lg\sigma$, thus $\Oh{n(H+1)}$ bits is $\Oh{n}$ space.  They also showed how to handle a variable $\tau$ and still achieve optimal query time, at the cost of increasing the space bound by a \((\log n)\)-factor.  That is, they gave a data structure that stores the string in $\Oh{n (H + 1)}$ words such that later, given the endpoints of a range and $\tau$, we can return the $\tau$-majorities of that range in $\Oh{1 / \tau}$ time. \citet{CDSW15} gave another solution for variable $\tau$, which also has $\Oh{1 / \tau}$ query time and uses $\Oh{n \log n}$ space.  As far as we know, these are all the relevant bounds for \citeauthor{KN08}'s original problem, both for fixed and variable $\tau$; they are summarized in Table~\ref{tab:results} together with our new results.  Related work includes dynamic structures \cite{EHMN11}, approximate solutions \cite{LPS08,WY11}, and encodings that do not access the string at query time \cite{NT15}. 

\begin{table}
\tbl{Results for the problem of parameterized range majority on a string of length $n$ over an alphabet $[1..\sigma]$ with $\sigma \le n$ in which the distribution of the elements has entropy $H \le \lg\sigma$. Note that all the spaces given in bits
are in $\Oh{n}$ words.
\label{tab:results}}
{\begin{tabular}{l@{\hspace{2ex}}|@{\hspace{2ex}}c@{\hspace{3ex}}c@{\hspace{3ex}}c}
source & space & time & $\tau$ is\\[.5ex]
\hline\\[-2ex]
\cite{KN08} & $\Oh{n (1 / \tau)}$ words & $\Oh{(1/\tau)(\log \log n)^2}$ & fixed\\[1ex]
\cite{DHMNS13} & $\Oh{{n \log (1 / \tau)}}$ words & $\Oh{1 / \tau}$ & fixed\\[1ex]
\cite{GHMN11} & $\Oh{n (H + 1)}$ bits & $\Oh{(1 / \tau) \log \log \sigma}$ & fixed\\[1ex]
%Theorem~\ref{thm:fixedtau} & $nH+o(n)(H+1)$ bits & $\Oh{1 / \tau}$ & fixed\\[1ex]
\cite{GHMN11} & $\Oh{n (H + 1)}$ words & $\Oh{1 / \tau}$ & variable\\[1ex]
\cite{CDSW15} & $\Oh{n \log n}$ words & $\Oh{1 / \tau}$ & variable\\[1ex]
\hline\\[-1.5ex]
Theorem~\ref{thm:compressed-smallsigma} ($\lg\sigma = \Oh{\lg w}$) & $nH+o(n)$ bits & $\Oh{1 / \tau}$ & variable\\[1ex]
Theorem~\ref{thm:compressed-largesigma} & $nH +o(n)(H+1)$ bits & Any $(1 / \tau)\cdot \omega(1)$ & variable\\[1ex]
Theorem~\ref{thm:epsilon-largesigma} & $(1+\epsilon)nH+o(n)$ bits & $\Oh{1 / \tau}$ & variable\\[1ex]
\end{tabular}}
\end{table}

In this paper we first consider the complementary problem of parameterized range minority, which was introduced by \citet{CDSW15} (and then generalized to trees by \citet{DSST16}).  For this problem we are asked to preprocess the string such that later, given the endpoints of a range, we can return (if one exists) a distinct element that occurs in that range but is not one of its $\tau$-majorities.  Such an element is called a $\tau$-minority for the range.  At first, finding a $\tau$-minority might seem harder than finding a $\tau$-majority because, for example, we are less likely to find a $\tau$-minority by sampling.  Nevertheless, \citeauthor{CDSW15} gave an $\Oh{n}$ space solution with $\Oh{1 / \tau}$ query time even for variable $\tau$.  In Section~\ref{sec:minority} we give two warm-up results, also for variable $\tau$:

\begin{enumerate}[(a)]
\item $\Oh{1 / \tau}$ query time using \((1+\epsilon)n H+\Oh{n}\) bits of space for any constant $\epsilon>0$; and
\item any query time of the form $(1 / \tau)\cdot\omega(1)$ using \(n H+o(n)(H+1)\) bits of space.
\end{enumerate}

That is, our spaces are not only linear like \citeauthor{CDSW15}$\!$'s, but also
compressed: we use either nearly optimally compressed space with no slowdown,
or optimally compressed space with nearly no slowdown. Within this space, we
can access $S$ in constant time. We reuse ideas from this section in our solutions for parameterized range majority.

In Section~\ref{sec:smallsigma} we return to \citeauthor{KN08}'s original problem of parameterized range majority, but with variable $\tau$, and consider the case where the alphabet of the string is polynomial on the computer word size, $\lg\sigma = \Oh{\lg w}$. We start with a simple linear-space structure (i.e., $\Oh{n\lg\sigma}$ bits, which is stricter than other linear-space solutions using $\Oh{n\lg n}$ bits) with worst-case optimal $\Oh{1 / \tau}$ query time, and then refine it to obtain optimally compressed space. In Section~\ref{sec:largesigma} we extend this solution to the more challenging case of larger alphabets, $\lg\sigma=\omega(\lg w)$. We also start with a linear-space structure and then work towards compressing it. While we obtain nearly optimally compressed space and worst-case optimal query time, reaching optimally compressed space imposes a slight slowdown in the query time. Summarizing, we obtain the following tradeoffs:
\begin{enumerate}
\item $\Oh{1 / \tau}$ query time using $nH+o(n)$ bits of space, where $\lg\sigma=\Oh{\lg w}$;
\item $\Oh{1 / \tau}$ query time using \((1+\epsilon)n H+o(n)\) bits of space for any constant $\epsilon>0$; and
\item any query time of the form $(1 / \tau)\cdot\omega(1)$ using \(n H+o(n)(H+1)\) bits of space.
\end{enumerate}
Overall, we obtain for the first time $\Oh{n}$-word (i.e., linear space) and 
$\Oh{1/\tau}$ (i.e., worst-case optimal) query time, even for variable $\tau$.
In all cases, we preserve constant-time access to $S$ within the compressed
space we use, and our queries require $\Oh{1/\tau}$ extra working space.
As a byproduct, we also show how to
%we show how, in the case of fixed $\tau$,
%we can have optimally compressed space, \(n H+o(n)(H+1)\) bits, and worst-case
%optimal query time, $\Oh{1/\tau}$, for any alphabet size. We also show how to
find the range mode in a time that depends on how frequent it actually is.

Finally, in Section~\ref{sec:minagain} we return to $\tau$-minorities. By
exploiting the duality with $\tau$-majorities, we reuse the results obtained
for the latter, so that the tradeoffs (1)--(3) are also obtained for 
$\tau$-minorities. Those results supersede the original tradeoff (a) obtained
in Section~\ref{sec:minority}. Actually, a single data structure with the
spaces given in points (1)--(3) solves at the same time $\tau$-minority and
$\tau$-majority queries.

\section{Preliminaries} \label{sec:preliminaries}

We use the RAM model of computation with word size in bits $w=\Omega(\lg n)$,
allowing multiplications. The input is an array $S[1..n]$ of
symbols (or ``elements'') from $[1..\sigma]$, where for simplicity we assume
$\sigma \le n$ (otherwise we could remap the alphabet so that every symbol
actually appears in $S$, without changing the output of any $\tau$-majority or
$\tau$-minority query).

\subsection{Access, select, rank, and partial rank} \label{subsec:queries}

Let \(S [1..n]\) be a string over alphabet $[1..\sigma]$, for $\sigma \le n$, and let $H \le \lg\sigma$ be the entropy of the distribution of elements in $S$, that is, $H=\sum_{a\in[1..\sigma]} \frac{n_a}{n}\lg\frac{n}{n_a}$, where each element $a$ appears $n_a$ times in $S$.  An access query on $S$ takes a position $k$ and returns \(S [k]\); a rank query takes an alphabet element $a$ and a position $k$ and returns $\rank_a(S,k)$, the number of occurrences of $a$ in \(S [1..k]\); a select query takes an element $a$ and a rank $r$ and returns $\select_a(S,r)$, the position of the $r$th occurrence of $a$ in $S$.  A partial rank query, $\rank_{S[k]}(S,k)$, is a rank query with the restriction that the element $a$ must occur in the position $k$.  These are among the most well-studied operations on strings, so we state here only the results most relevant to this paper.

For \(\sigma = 2\) and any constant $c$, \citet{Pat08} showed how we can store $S$ in \(n H + \Oh{n / \log^c n}\) bits, supporting all the queries in time $\Oh{c}$. If $S$ has $m$ 1s, this space is $\Oh{m\lg\frac{n}{m}+n / \log^c n}$ bits.
For \(\lg \sigma = \Oh{\log\log n}\), \citet{FMMN07} showed how we can store $S$ in \(n H + o (n)\) bits and support all the queries in $\Oh{1}$ time. This result was later extended to the case $\lg\sigma = \Oh{\lg w}$ \cite[Thm.\ 7]{BN15}. \citet{BCGNN13} showed how, for any positive constant $\epsilon$, we can store $S$ in \((1 + \epsilon) n H + o (n)\) bits and support access and select in $\Oh{1}$ time and rank in $\Oh{\log \log \sigma}$ time.  Alternatively, they can store $S$ in $nH+o(n)(H+1)$ bits and support either access or select in time $\Oh{1}$, and the other operation, as well as rank, in time $\Oh{\log\log\sigma}$. \citet[Thm.\ 8]{BN15} improved the time of rank to $\Oh{\lg\lg_w \sigma}$, which they proved optimal, and the time of the non-constant operation to any desired function in $\omega(1)$. \citet[Sec.\ 3]{BN14} showed how to support $\Oh{1}$-time partial rank using \(o (n) (H + 1)\) further bits. 
Throughout this article we also use some simpler variants of these results.

\subsection{Colored range listing} \label{subsec:listing}

\citet{Mut02} showed how we can store \(S [1..n]\) such that, given the endpoints of a range, we can quickly list the distinct elements in that range and the positions of their leftmost occurrences therein.  Let \(C [1..n]\) be the array in which \(C [k]\) is the position of the rightmost occurrence of the element \(S [k]\) in \(S [1..k - 1]\) --- i.e., the last occurrence before \(S [k]\) itself --- or 0 if there is no such occurrence.  Notice \(S [k]\) is the first occurrence of that distinct element in a range \(S [i..j]\) if and only if \(i \leq k \leq j\) and \(C [k] < i\).  We store $C$, implicitly or explicitly, and a data structure supporting $\Oh{1}$-time range-minimum queries on $C$ that returns the position of the leftmost occurrence of the minimum in the range.

To list the distinct elements in a range \(S [i..j]\), we find the position $m$ of the leftmost occurrence of the minimum in the range \(C [i..j]\); check whether \(C [m] < i\); and, if so, output \(S [m]\) and $m$ and recurse on \(C [i..m - 1]\) and \(C [m + 1..j]\).  This procedure is online --- i.e., we can stop it early if we want only a certain number of distinct elements --- and the time it takes per distinct element is $\Oh{1}$ plus the time to access $C$.

Suppose we already have data structures supporting access, select and partial rank queries on $S$, all in $\Oh{t}$ time.  Notice \(C [k] = \select_{S [k]} (S, \rank_{S [k]} (S,k) - 1)\), so we can also support access to $C$ in $\Oh{t}$ time.  \citet{Sad07} and \citet{Fis10} gave $\Oh{n}$-bit data structures supporting $\Oh{1}$-time range-minimum queries.  Therefore, we can implement \citeauthor{Mut02}'s solution using $\Oh{n}$ extra bits such that it takes $\Oh{t}$ time per distinct element listed.

\section{Parameterized Range Minority} \label{sec:minority}

\citet{CDSW15} gave a linear-space solution with $\Oh{1 / \tau}$ query time for parameterized range minority, even for the case of variable $\tau$ (i.e., chosen at query time).  They first build a list of \(\lceil 1 / \tau \rceil\) distinct elements that occur in the given range (or as many as there are, if fewer) and then check those elements' frequencies to see which are $\tau$-minorities.  There cannot be as many as \(\lceil 1 / \tau \rceil\) $\tau$-majorities so, if there exists a $\tau$-minority for that range, then at least one must be in the list.  In this section we use a simple approach to implement this idea using compressed space; then we obtain more refined results in Section~\ref{sec:minagain}.

\subsection{Compressed space}

To support parameterized range minority on \(S [1..n]\) in $\Oh{1 / \tau}$ time, we store a data structure supporting $\Oh{1}$-time access, select and partial rank queries on $S$ \cite{BCGNN13} and a data structure supporting $\Oh{1}$-time range-minimum queries on the $C$ array corresponding to $S$ \cite{Mut02}. As seen in Section~\ref{sec:preliminaries}, for any positive constant $\epsilon$, we can store these data structures in a total of \((1 + \epsilon) n H + \Oh{n}\) bits.  Given $\tau$ and endpoints $i$ and $j$, in $\Oh{1 / \tau}$ time we use \citeauthor{Mut02}'s algorithm to build a list of \(\lceil 1 / \tau \rceil\) distinct elements that occur in \(S [i..j]\) (or as many as there are, if fewer) and the positions of their leftmost occurrences therein.  We check whether these distinct elements are $\tau$-minorities using the following lemma:

\begin{lemma} \label{lem:check}
Suppose we know the position of the leftmost occurrence of an element in a range. Then we can check whether that element is a $\tau$-minority or a $\tau$-majority using a partial rank query and a select query on $S$.
\end{lemma}

\begin{proof}
Let $k$ be the position of the first occurrence of $a$ in \(S [i..j]\).  If \(S [k]\) is the $r$th occurrence of $a$ in $S$, then $a$ is a $\tau$-minority for \(S [i..j]\) if and only if the \((r + \lfloor \tau (j - i + 1) \rfloor)\)th occurrence of $a$ in $S$ is strictly after \(S [j]\); otherwise $a$ is a $\tau$-majority.  That is, we can check whether $a$ is a $\tau$-minority for \(S [i..j]\) by checking whether
\[\select_a (S, \rank_a (S,k) + \lfloor \tau (j - i + 1) \rfloor ) > j\,;\]
since \(S [k] = a\), computing \(\rank_a (S,k)\) is only a partial rank query.
\qed
\end{proof}

This gives us the following theorem, which improves \citeauthor{CDSW15}$\!$'s solution to use nearly optimally compressed space with no slowdown.

\begin{theorem} \label{thm-min:epsilon}
Let $S[1..n]$ be a string whose distribution of symbols has entropy $H$.
For any constant $\epsilon>0$, we can store $S$ in \((1 + \epsilon) n H + \Oh{n}\) bits such that later, given the endpoints of a range and $\tau$, we can return a $\tau$-minority for that range (if one exists) in time $\Oh{1 / \tau}$.
\end{theorem}

\subsection{Optimally compressed space} \label{sec-min:compressed}

By changing our string representation to that of \citet[Thm.\ 8]{BN15},
we can store our data structures for access, select and partial rank on $S$ and range-minimum queries on $C$ in a total of \(n H + o(nH) + \Oh{n}\) bits at the cost of the select queries taking $\Oh{g (n)}$ time, for any desired \(g (n) = \omega (1)\); see again Section~\ref{sec:preliminaries}. 
Therefore the range minority is found in time $\Oh{(1/\tau) g(n)}$.

%\begin{theorem} \label{thm:small min}
%For any function \(f (n) = \omega (1)\), we can store $S$ in \(n H + o(nH)+\Oh{n}\) bits such that later, given the endpoints of a range and $\tau$, we can return a $\tau$-minority for that range (if one exists) in $\Oh{(1 / \tau)\,f (n)}$ time.
%\end{theorem}

To reduce the space bound to \(n H + o (n) (H + 1)\) bits, we must reduce the space of the range-minimum data structure to $o(n)$. Such a result was sketched by \citet{HSV09}, but it lacks sufficient detail to ensure correctness. We give these details next.

The technique is based on sparsification. We cut the sequence into blocks of length $g(n)$, choose the $n/g(n)$ minimum values of each block, and build the range-minimum data structure on the new array $C'[1..n/g(n)]$ (i.e., $C'[i]$ stores the minimum 
of $C[(i-1)\cdot g(n)+1..i\cdot g(n)]$). This requires $\Oh{n/g(n)} = o(n)$ bits.
\citeauthor{Mut02}'s algorithm is then run over $C'$ as follows. We find the minimum
position in $C'$, then recursively process its left interval, then process the
minimum of $C'$ by considering the $g(n)$ corresponding cells in $C$, and 
finally process the right part of the interval. The recursion stops when the
interval becomes empty or when all the $g(n)$ elements in the block of $C$
are already reported. 

\begin{lemma} \label{lem:sparsemuthu}
The procedure described identifies the leftmost positions of all the
distinct elements in a block-aligned interval $S[i..j]$, working over at most 
$g(n)$ cells per new element discovered.
\end{lemma}
\begin{proof}
We proceed by induction on the size of the current subinterval $[\ell..r]$,
which is always block-aligned.
%The result is trivial for the empty interval. 
Let $k'$ be the position of the minimum in $C'[(\ell-1)g(n)+1..r/g(n)]$
and let $k$ be the position of the minimum in $C[(k'-1)\cdot g(n)+1,k'\cdot g(n)]$. 
Then $C[k]$ is clearly the minimum in $C[\ell..r]$ and $S[k]$ is the leftmost occurrence in $S[\ell..r]$ of 
the element $a=S[k]$. If $C[k] \ge i$, then 
$a$ already occurs in $S[i..\ell-1]$ and we have already reported it. Since
the minimum of $C[\ell..r]$ is within the block $k'$ of $C$, it is sufficient
that $C[k] \ge i$ for all the positions $k$ in that block to ensure that all
the values in $S[\ell..r]$ have already been reported, in which case we can 
stop the procedure. The $g(n)$ scanned cells can be charged to the function that
recursively invoked the interval $[\ell..r]$.

Otherwise, we recursively process the interval to the left of
block $k'$, which by inductive hypothesis reports the unique elements in that
interval. Then we process the current block of size $g(n)$, finding at least
the new occurrence of element $S[k]$ (which cannot appear to the left of $k'$).
Finally, we process the interval
to the right of $k'$, where the inductive hypothesis again holds. 

Note that the method is also correct if, instead of checking whether all the
elements in the block of $C'[k']$ are $\ge i$, we somehow check that all of 
them have already been reported. We will use this variant later in the paper.
\qed
\end{proof}

For general ranges $S[i..j]$, we must include in the range of $C'$ the two 
partially overlapped blocks on the extremes of the range. When it comes to 
process one of those blocks, we only consider the cells that are inside 
$[i..j]$; the condition to report an element is still that $C[k] < i$. 

We use this procedure to obtain any $\lceil 1/\tau\rceil$ distinct elements.
We perform $\Oh{(1/\tau)\,g(n)}$ accesses to $C$, each of which costs time
$\Oh{g(n)}$ because it involves a select query on $S$. Therefore the total time
is $\Oh{(1/\tau)\,g(n)^2}$. Testing each of the candidates with 
Lemma~\ref{lem:check} takes time $\Oh{(1/\tau)\,g(n)}$, because we also use 
select queries on $S$. Therefore, for
any desired time of the form $\Oh{(1/\tau)\,f(n)}$, we use $g(n)=\sqrt{f(n)}$.

\begin{theorem} \label{thm-min:compressed}
Let $S[1..n]$ be a string whose distribution of symbols has entropy $H$.
For any function \(f (n) = \omega (1)\), we can store $S$ in \(n H + o(n)(H+1)\) bits such that later, given the endpoints of a range and $\tau$, we can return a $\tau$-minority for that range (if one exists) in time $\Oh{(1/\tau)\,f (n)}$.
\end{theorem}

Note that this representation retains constant-time access to $S$.

\section{Parameterized Range Majority on Small Alphabets} 
\label{sec:smallsigma}

In this section we consider the case $\lg\sigma = \Oh{\lg w}$, where rank 
queries on $S$ can be supported in constant time. Our strategy is to find a
set of $\Oh{1/\tau}$ candidates that contain all the possible $\tau$-majorities
and then check them one by 
one, counting their occurrences in $S[i..j]$ via rank queries on $S$. The 
time will be worst-case optimal, $\Oh{1/\tau}$. We first obtain 
$\Oh{n\lg\sigma}$ bits of space and then work towards compressing it.

First, note that if $\tau < 1/\sigma$, we can simply assume that all the 
$\sigma$ symbols are candidates for majority, and check them one by
one; therefore we care only about how to find $\Oh{1/\tau}$ candidates in the 
case $\tau \ge 1/\sigma$. 

\subsection{Structure} \label{sec:smallsigma-struc}

We store an instance of the structure of \citet[Thm.\ 5]{BN15} supporting access, rank, and select on $S$ in $\Oh{1}$ time, using $n\lg\sigma+o(n)$ bits. For
every \(0 \leq t \leq \lceil \log \sigma \rceil\) and \(t \leq b \leq \lfloor \log n \rfloor\), we divide $S$ into blocks of length $2^{b - 1}$ and store a binary string \(G_b^t [1..n]\) in which \(G_b^t [k] = 1\) if, (1) the element \(S [k]\) occurs at least $2^{b - t}$ times in \(S [k - 2^{b + 1}..k + 2^{b + 1}]\), and (2) \(S [k]\) is the leftmost or rightmost occurrence of that element in its block. 

At query time we will use \(t = \lceil \log (1 / \tau) \rceil\) and \(b = \lfloor \log (j - i + 1) \rfloor\). The following lemma shows that it is sufficient to
consider the candidates $S[k]$ for $i \le k \le j$ where $G_b^t[k]=1$.

\begin{lemma}
For every $\tau$-majority $a$ of $S[i..j]$ there exists some $k \in [i..j]$
such that $S[k]=a$ and $G_b^t[k]=1$.
\end{lemma}
\begin{proof}
Since \(S [i..j]\) cannot be completely contained in a block of length $2^{b - 1}$, if \(S [i..j]\) overlaps a block then it includes one of that block's endpoints.  Therefore, if \(S [i..j]\) contains an occurrence of an element $a$, then it includes the leftmost or rightmost occurrence of $a$ in some block.  Suppose $a$ is a $\tau$-majority in \(S [i..j]\), and \(b \geq t\).  For all \(i \leq k \leq j\), $a$ occurs at least $\tau 2^b \ge 2^{b - t}$ times in \(S [k - 2^{b + 1}..k + 2^{b + 1}]\), so since some occurrence of $a$ in \(S [i..j]\) is the
leftmost or rightmost in its block, it is flagged by a 1 in $G_b^t[i..j]$.
\end{proof}

The number of distinct elements that occur at least $2^{b - t}$ times in a range of size $\Oh{2^b}$ is $\Oh{2^t}$, so in each block there are $\Oh{2^t}$ positions flagged by 1s in $G_b^t$, for a total of $m=\Oh{n\,2^{t-b}}$ 1s.  It follows that we can store an instance of the structure of \citet{Pat08} (recall
Section~\ref{subsec:queries}) supporting $\Oh{1}$-time access, rank and select on $G_b^t$ in $\Oh{n 2^{t - b} (b - t) + n / \log^3 n}$ bits in total.  Summing over $t$ from 0 to \(\lceil \log \sigma \rceil\) and over $b$ from \(t\) to \(\lfloor \log n \rfloor\), calculation shows we use a total of \(\Oh{n \log \sigma}\) bits for the binary strings. 

\subsection{Queries} \label{sec:smallsigma-query}

Given endpoints $i$ and $j$ and a threshold $\tau$, if $\tau < 1/\sigma$, we
simply report every element $a \in [1..\sigma]$ such that $\rank_a(S,j)-
\rank_a(S,i-1) > \tau(j-i+1)$, in total time $\Oh{\sigma} = \Oh{1/\tau}$. 
Otherwise, we compute $b$ and $t$ as explained and, if $b < t$, we
run a sequential algorithm on \(S [i..j]\) in \(\Oh{j - i} = \Oh{1 / \tau}\) time \cite{MG82}. Otherwise, we use rank and select on $G_b^t$ to find all the 1s in $G_b^t[i..j]$. Since \(S [i..j]\) overlaps at most 5 blocks of length $2^{b - 1}$, it contains $\Oh{1 / \tau}$ elements flagged by 1s in $G_b^t$; therefore, we have $\Oh{1/\tau}$ candidates to evaluate, and these include all the
possible $\tau$-majorities. Each candidate $a$ is tested in constant
time for the condition $\rank_a(S,j)-\rank_a(S,i-1) > \tau(j-i+1)$.

We aim to run the sequential algorithm in $\Oh{j-i}=\Oh{1/\tau}$ worst-case 
time and space, which we achieve by taking advantage of the rank and select 
operations on $S$. We create a doubly-linked list with the positions $i$ to 
$j$, plus an array $T[1..j-i+1]$ where $T[k]$ points to the list node 
representing $S[i+k-1]$. We take the element $S[i]=a$ at 
the head of the list and know that it is a $\tau$-majority in $S[i..j]$ if 
$\rank_a(S,j)-\rank_a(S,i-1) > \tau(j-i+1)$. If it is, we immediately report 
it. In any case, we remove all the occurrences of $a$ from the doubly-linked 
list, that is, the list nodes $T[\select_a(S,\rank_a(S,i)+r)]$, $r=0,1,2,\ldots$.
We proceed with the new header of the doubly-linked list, which points to a 
different element $S[i']=a'$, and so on. It is clear that we perform $\Oh{j-i}$
constant-time rank and select operations on $S$, and that at the end we have 
found all the $\tau$-majorities.

\begin{theorem} \label{thm:linear-smallsigma}
Let $S[1..n]$ be a string over alphabet $[1..\sigma]$, with 
$\lg\sigma=\Oh{\lg w}$. We
can store $S$ in \(\Oh{n\lg\sigma}\) bits such that later, given the endpoints 
of a range and $\tau$, we can return the $\tau$-majorities for that range in 
time $\Oh{1 / \tau}$.
\end{theorem}
% the space constant appears to be 16

%\begin{theorem} \label{thm:compr-slow-largesigma}
%Given a string $S[1..n]$ with zero-order entropy $H$, 
%over alphabet $[1..\sigma]$, with $\lg \sigma=\omega(\lg w)$, 
%we can store $S$ in \(n H + o (n \log \sigma)\) bits such that later, given the endpoints of a range and $\tau$, we can return the $\tau$-majorities for that range in $\Oh{(1 / \tau) \log \log_w \sigma}$ time.
%\end{theorem}

\subsection{Succinct space} \label{sec:smallsigma-succinct}

To reduce the space we will open the structure we are using to represent $S$ 
\cite[Thm.\ 5]{BN15}. This is a multiary 
wavelet tree: it cuts the alphabet range $[1..\sigma]$ into $w^\beta$ contiguous
subranges of about the same size, for some conveniently small constant 
$0<\beta\le 1/4$. The root node $v$ of the wavelet tree stores the sequence
$S_v[1..n]$ indicating the range to which each symbol of $S$ belongs, roughly
$S_v[i] = \lceil S[i]/(\sigma/w^\beta)\rceil$. This node has $w^\beta$ 
children, where the $p$th child stores the subsequence of the symbols $S[i]$ 
such that $S_v[i]=p$. The alphabet of each child has been reduced to a range 
of size roughly $\sigma / w^\beta$. This range is split again into $w^\beta$ 
subranges, creating $w^\beta$ children for each child, and so on. The process 
is repeated recursively until the alphabet range is of size less than 
$w^\beta$. The wavelet tree has height $\log_{w^\beta} \sigma = 
\frac{\lg \sigma}{\beta\lg w} =\Oh{1}$, and at each level the strings $S_v$ 
stored add up to $n\lg(w^\beta) = \beta n \lg w$ bits, for a total of
$n\lg\sigma$ bits of space. The other $o(n)$ bits are needed to provide 
constant-time rank and select support on the strings $S_v$.

We use this hierarchical structure to find the $\tau$-majorities as follows.
Assume we have the structures to find $\tau$-majorities in any of the strings
$S_v$ associated with wavelet tree nodes $v$, in time $\Oh{1/\tau}$. 
Then, if $a$ is a $\tau$-majority in $S[i..j]$, the symbol 
$p=\lceil a/(\sigma/w^\beta)\rceil$ is also a $\tau$-majority in $S_v[i..j]$, 
where $v$ is the wavelet tree root. Therefore, we find in time $\Oh{1/\tau}$ 
the $\tau$-majorities $p$ in $S_v[i..j]$.
We verify each such $\tau$-majority $p$ recursively in the $p$th child of $v$.
In this child $u$, the range $S_v[i..j]$ is projected to $S_u[i_u..j_u] = 
S_u[\rank_p(S_v,i-1)+1,\rank_p(S_v,j)]$, and the corresponding threshold is
$\tau_u = \tau (j-i+1)/(j_u-i_u+1)<1$. This process continues recursively until
we find the majorities in the leaf nodes, which correspond to
actual symbols that can be reported as $\tau$-majorities in $S[i..j]$.

The time to find the $\tau_u$-majorities in each child $u$ of the root $v$ is
$\Oh{1/\tau_u}=\Oh{(j_u-i_u+1)/((j-i+1)\tau)}$. Added over all the children
$u$, this gives $\sum_u \Oh{1/\tau_u} = \sum_u \Oh{(j_u-i_u+1)/((j-i+1)\tau)} = \Oh{1/\tau}$. Adding this over all the levels, we obtain
$\Oh{(1/\beta)(1/\tau)}=\Oh{1/\tau}$.

\paragraph{Finding $\tau'$-majorities on tiny alphabets}

The remaining problem is how to find $\tau'$-majorities on an alphabet of size
$\sigma'=w^\beta$, on each of the strings $S_v$ of length $n_v$. We do almost 
as we did for Theorem~\ref{thm:linear-smallsigma}, except that the range for $b$ is 
slightly narrower: \(\lfloor \log (2^t \cdot w^\beta/4) \rfloor \leq b \leq 
\lfloor \log n_v \rfloor\).
Then calculation shows that the total space for the bitvectors $G_b^t$ is 
\(\Oh{\frac{n_v \log \sigma' \lg w}{w^\beta} + \frac{n_v}{\log n_v}} = 
o(n_v)\), so added over the whole wavelet tree is $o(n)$. 

The price of using this higher lower bound for $b$ is that it requires us to 
sequentially find $\tau'$-majorities in time $\Oh{1/\tau'}$ on ranges of length
$\Oh{(1/\tau') w^\beta}$. However, we can take advantage of the small alphabet. 
First, if $1/\tau' \ge \sigma'$, we just perform $\sigma'$ pairs of 
constant-time rank queries on $S_v$. For $1/\tau' < \sigma'$, we will compute
an array of $\sigma'$ counters with the frequency of the symbols in the range, 
and then report those exceeding the threshold. The maximum
size of the range is $(4/\tau')w^\beta/4 \le \sigma' w^\beta=w^{2\beta}$, and
thus $2\beta\lg w$ bits suffice to represent each counter. The $\sigma'$ 
counters then require $2\beta w^\beta\lg w$ bits and can be maintained in a 
computer word (although we will store them somewhat spaced for technical
reasons). We can read the elements in $S_v$ by chunks of $w^\beta$ 
symbols, and compute in constant time the corresponding
counters for those symbols. Then we sum the current counters and the 
counters for the chunk, all in constant time because they are fields in a 
single computer word. The range is then processed in time $\Oh{1/\tau'}$. 

To compute the counters corresponding to $w^\beta$ symbols, we extend the
popcounting algorithm of \citet[Sec.\ 4.1]{BN15}. Assume we extract them from
$S_v$ and have them packed in the lowest $k\ell$ bits of a computer word $X$, 
where $k=w^\beta$ is the number of symbols and $\ell=\lg\sigma'$ the number of
bits used per symbol. We first create $\sigma'$ copies of the sequence at 
distance $2k\ell$ of each other:
$X \leftarrow X \cdot (0^{2k\ell-1}1)^{\sigma'}$. 
In each copy we will count the occurrences of a different symbol.
To have the $i$th copy count the occurrences of symbol $i$, for $0 \le i <
\sigma'$, we perform
$$X ~\leftarrow~ X~~\textsc{xor}~~0^{k\ell} ((\sigma'-1)_\ell)^k \ldots 0^{k\ell} (2_\ell)^k ~ 0^{k\ell} (1_\ell)^k ~ 0^{k\ell} (0_\ell)^k,$$
where $i_\ell$ is number $i$ written in $\ell$ bits.
Thus in the $i$th copy the symbols equal to $i$ become zero and the others
nonzero. To set a 1 at the highest bit of the symbols equal to $i$ in the $i$th
copy, we do
$$X ~\leftarrow~ (Y - (X~\textsc{and~not}~Y))~\textsc{and}~Y~\textsc{and~not}~X,$$ 
where $Y=(0^{k\ell} (10^{\ell-1})^k)^{\sigma'}$.%
\footnote{This could have been simply $X \leftarrow (Y-X)~\textsc{and}~Y$ if 
there was an unused highest bit set to zero in the fields of $X$. Instead,
we have to use this more complex formula that first zeroes the 
highest bit of the fields and later considers them separately.}
Now we add all the 1s in each copy with
$X \leftarrow X \cdot 0^{k\ell(2\sigma'-1)} (0^{\ell-1}1)^k$.
This spreads several sums across the $2k\ell$ bits of each copy, and in 
particular the $k$th sum adds up all the 1s of the copy. Each sum requires 
$\lg k$ bits, which is precisely the $\ell$ bits we have allocated per field.
Finally, we isolate the desired counters using
$X \leftarrow X~\textsc{and}~(0^{k\ell}1^\ell0^{(k-1)\ell})^{\sigma'}$.
The $\sigma'$ counters are not contiguous in the computer word, but we still
can afford to store them spaced: we use $2k\ell\sigma' = 
2 \beta w^{2\beta}\lg w$ bits, which since $\beta \le 1/4$, is always less 
than $w$.

The cumulative counters, as said, need $\lg (\sigma' w^\beta) = 2\ell$
bits. We will store them in a computer word $C$ separated 
by $2k\ell$ bits so that we can directly add the resulting word $X$ after 
processing a chunk of $w^\beta$ symbols of the range in $S_v$: 
$C \leftarrow C+X$. If the last chunk is of length $l < w^\beta$, we complete
it with zeros and then subtract those spurious $w^\beta-l$ occurrences from the
first counter, 
$C \leftarrow C - (w^\beta-l) \cdot 2^{(k-1)\ell}$.

The last challenge is to output the counters that are at least 
$y = \lfloor \tau' (j-i+1) \rfloor+1$ after processing the range. We
use 
$$C \leftarrow C + (2^{2\ell}-y) \cdot (0^{k\ell+\ell-1} 1 0^{(k-1)\ell})^{\sigma'}$$
so that the counters reaching $y$ will overflow to the next bit. We isolate
those overflow bits with
$C \leftarrow C~\textsc{and}~(0^{(k-1)\ell-1} 1 0^{(k+1)\ell})^{\sigma'}$,
so that we have to report the $i$th symbol if and only if 
$C~\textsc{and}~0^{(k(2\sigma'-2i+1)-1)\ell-1} 1 0^{(k(2i-1)+1)\ell} \not= 0$.
We repeatedly isolate the lowest bit of $C$ with 
$$D \leftarrow (C~\textsc{xor}~(C-1))~\textsc{and}~
(0^{(k-1)\ell-1} 1 0^{(k+1)\ell})^{\sigma'},$$
and then remove it with $C \leftarrow C~\textsc{and}~(C-1)$, until $C=0$.
Once we have a position isolated in $D$, we find the position in constant
time by using a monotone minimum perfect hash function over the set 
$\{ 2^{(k(2i-1)+1)\ell},~1 \le i \le \sigma'\}$, which uses
$\Oh{\sigma' \lg w}=o(w)$ bits \cite{BBPV09}. Only one such data structure
is needed for all the sequences, and it takes less space than a single 
systemwide pointer.

\begin{theorem} \label{thm:succinct-smallsigma}
Let $S[1..n]$ be a string over alphabet $[1..\sigma]$, with 
$\lg\sigma=\Oh{\lg w}$. We can store $S$ in \(n \lg\sigma + o(n)\) bits such 
that later, given the endpoints of a range and $\tau$, we can return the 
$\tau$-majorities for that range in time $\Oh{1/\tau}$.
\end{theorem}

\subsection{Optimally compressed space} \label{sec:smallsigma-compressed}

One choice to compress the space is to use a compressed representation of
the strings $S_v$ \cite{FMMN07}. This takes chunks of $c=(\lg n)/2$ bits
and assigns them a code formed by a header of $\lg c$ bits and a variable-length
remainder of at most $c$ bits. For 
decoding a chunk in constant time, they use a directory of 
$\Oh{n_v\lg c/c}$ bits, plus a constant table of size $2^c=\Oh{\sqrt{n}}$ that 
receives any encoded string and returns the original chunk.
The compressed size of any string $S_v$ with entropy $H_v$ then 
becomes $n_vH_v+ \Oh{n_v\log\sigma'\log\log n/\log n}$ bits, which added over
the whole wavelet tree becomes $nH+\Oh{\frac{n\log\sigma\log\log n}{\log n}}$
bits.
This can be used in replacement of the direct representation of 
sequences $S_v$ in Theorem~\ref{thm:succinct-smallsigma}, since we only change
the way a chunk of $\Theta(\lg n)$ bits is read from any $S_v$. Note that we
read chunks of $w^\beta$ symbols from $S_v$, which could be $\omega(\lg n)$ if
$n$ is very small. To avoid this problem, we apply this method only when
$\lg\sigma = \Oh{\lg\lg n}$, as in this case we can use computer words of
$w=\lg n$ bits.

\begin{corollary} \label{cor:compressed-smallsigma}
Let $S[1..n]$ be a string whose distribution of symbols has entropy $H$, over
alphabet $[1..\sigma]$, with $\lg\sigma=\Oh{\lg\lg n}$. We can store $S$ in 
\(nH+ o(n)\) bits such that later, given the endpoints of a range and $\tau$, 
we can return the $\tau$-majorities for that range in time $\Oh{1/\tau}$.
\end{corollary}

% GZL: To avoid depending on tables here as well, we use AP instead
%By using instead $c=\alpha f(n) \lg w$, we can extract any chunk 
%of $\alpha f(n) \lg w$ bits in constant time, and the total space 
%is $nH+\Oh{\frac{n\log\sigma (\log f(n)+\lg\lg w)}{f(n)\lg w}} =
%nH+o(n)$ bits, plus the $w^{\alpha f(n)}$ space for the constant table.

For the case where $\lg\sigma=\omega(\lg\lg n)$ but still $\lg\sigma =
\Oh{\lg w}$, we use another technique. We represent $S$ using the optimally
compressed structure of \citet{BCGNN13}. This structure separates the alphabet 
symbols into $\lg^2 n$ classes according to their frequencies. A sequence 
$K[1..n]$, where $K[i]$ is the class to which $S[i]$ is assigned, is represented
using the structure of Corollary~\ref{cor:compressed-smallsigma}, which 
supports constant-time access, rank, and select, since the alphabet of $K$ is 
of polylogarithmic size, and also $\tau$-majority queries in time 
$\Oh{1/\tau}$. For each class $c$, a sequence $S_c[1..n_c]$ contains the 
subsequence of $S$ of the symbols $S[i]$ where $K[i]=c$.
We will represent the subsequences $S_c$ using
Theorem~\ref{thm:succinct-smallsigma}. Then the structure for $K$ takes 
$nH_K+o(n)$ bits, where $H_K$ is the entropy of the distribution of the
symbols in $K$, and the structures for the strings $S_c$ take 
$n_c\lg\sigma_c + o(n_c)$ bits, where $S_c$ ranges over alphabet 
$[1..\sigma_c]$. \citeauthor{BCGNN13} show that these space bounds
add up to $nH + o(n)$ bits and that one can support access, rank and select on 
$S$ via access, rank and select on $K$ and some $S_c$.

Our strategy to solve a $\tau$-majority query on $S[i..j]$ resembles the one
used to prove Theorem~\ref{thm:succinct-smallsigma}.
We first run a $\tau$-majority query on string $K$. This will yield
the at most $1/\tau$ classes of symbols that, together, occur more than
$\tau(j-i+1)$ times in $S[i..j]$. The classes excluded from this result cannot
contain symbols that are $\tau$-majorities. Now, for each included class $c$,
we map the interval $S[i..j]$ to $S_c[i_c..j_c]$ in the subsequence of its
class, where $i_c = \rank_c(K,i-1)+1$ and $j_c = \rank_c(K,j)$, and then 
run a $\tau_c$-majority query on $S_c[i_c..j_c]$, for 
$\tau_c = \tau(j-i+1)/(j_c-i_c+1)$. The results obtained for each 
considered class $c$ are reported as $\tau$-majorities in $S[i..j]$.
The query time, added over all the possible $\tau_c$ values, is 
$\sum_c \Oh{1/\tau_c} = \Oh{1/\tau}$ as before.

\begin{theorem} \label{thm:compressed-smallsigma}
Let $S[1..n]$ be a string whose distribution of symbols has entropy $H$, over
alphabet $[1..\sigma]$, with $\lg\sigma=\Oh{\lg w}$. We can store $S$ in 
\(n H + o(n)\) bits such that later, given the endpoints of a range and $\tau$,
we can return the $\tau$-majorities for that range in time $\Oh{1/\tau}$.
\end{theorem}

%Note that this structure can still be applied if $\lg\sigma=\omega(\lg w)$,
%but the query time would increase proportionally to the height of the wavelet 
%tree, to $\Oh{(1/\tau)\lg_w \sigma}$, and the space would raise to
%$nH+o(nH)$ bits. More attractive results are obtained next.
%GZL: Not interesting anymore with that space

\section{Parameterized Range Majority on Large Alphabets}
\label{sec:largesigma}

Rank queries cannot be performed in constant time on large alphabets
\cite{BN15}.
To obtain optimal query time in this case, we resort to the use of 
Lemma~\ref{lem:check} instead of performing rank queries on $S$. For 
this purpose, we must
be able to find the leftmost occurrence of each $\tau$-majority in a range. 
This is done by adding further structures on top of the bitvectors $G_b^t$ used
in Theorem~\ref{thm:linear-smallsigma}. Those bitvectors $G_b^t$ alone require 
$\Oh{n\lg\sigma}$ bits of space, whereas our further structures add only $o(n)$
bits. Within these $\Oh{n\lg\sigma}$ bits, we can store a simple representation 
of $S$ \cite{BCGNN13}, which supports both access 
and select queries in constant time. We also add the structures
to support partial rank in constant time, within $o(n\lg\sigma)$ further
bits. Therefore we can apply Lemma~\ref{lem:check} in constant time and solve
$\tau$-majority queries in time $\Oh{1/\tau}$.
We consider compression later.

\subsection{Structure}

First, to cover the case $\tau < 1/\sigma$, we build the structure of 
\citet{Mut02} on $S$, using $\Oh{n}$ extra bits as shown in 
Section~\ref{subsec:listing}, so that we can find the $\Oh{\sigma}=\Oh{1/\tau}$
leftmost occurrences of each distinct element in $S[i..j]$. On each leftmost
occurrence we can then apply Lemma~\ref{lem:check} in constant time. Now we 
focus on the case $\tau \ge 1/\sigma$. 

In addition to the bitvectors $G_b^t$ of the previous section,
we mark in a second bitvector $J_b^t$ each $(\lg^4 n)$-th occurrence in 
$S$ of the alphabet symbols in the area where they mark bits in $G_b^t$. 
More precisely, let $a=S[k]$ occur at least $2^{b-t}$ times in 
$S[k-2^{b+1}..k+2^{b+1}]$, and let $i_1, i_2, \ldots$ be the positions of 
$a$ in $S$. Then we mark in $J_b^t$ the positions $\{ i_{q\lg^4 n}, k-2^{b+1} 
\le i_{q\lg^4 n} \le k+2^{b+1} \}$.

For the subsequence $S_b^t$ of elements of $S$ marked in $J_b^t$, we build an 
instance of \citeauthor{Mut02}'s structure. 
That is, we build the structure on the array $C_b^t$ corresponding to the 
string $S_b^t[k]=S[\select_1(J_b^t,k)]$. This string need not be stored 
explicitly, but instead we store $C_b^t$ in explicit form.

Furthermore, if for any $b$ and $t$ it holds $J_b^t[i_{q\lg^4 n}]=1$, being 
$S[i_{q\lg^4 n}]=a$, we create a succinct SB-tree \cite[Lem~3.3]{GRR09} 
successor structure\footnote{In that paper they find predecessors, but the 
problem is analogous.} associated with the chunk of $\lg^4 n$ consecutive 
positions of $a$: $i_{1+(q-1)\lg^4 n}, \ldots, i_{q\lg^4 n}$. This structure is 
stored associated with the 1 at $J_b^t[i_{q\lg^4 n}]$ (all the 1s at the
same position $i_{q\lg^4 n}$, for different $b$ and $t$ values, point to
the same succinct SB-tree, as it does not depend on $b$ or $t$). 
The SB-tree operates
in time $\Oh{\lg(\lg^4 n)/\lg\lg n}=\Oh{1}$ and uses $\Oh{\lg^4 n \lg\lg n}$
bits. It needs constant-time access to the positions $i_{r+(q-1)\lg^4 n}$, as it
does not store them. We provide those positions using $i_k=\select_a(S,k)$. 

Added over all the symbols $a$, occurring $n_a$ times in $S$, each bitvector 
$J_b^t$ contains $\sum_a \lfloor n_a/\lg^4 n \rfloor = \Oh{n/\lg^4 n}$ 1s. 
Thus, added over every $b$ and $t$, the bitvectors $J_b^t$, arrays $C_b^t$,
and pointers to succinct SB-trees (using $\Oh{\lg n}$ bits per pointer),
require $\Oh{n/\lg n} = o(n)$ bits. 
Each succinct SB-tree requires $\Oh{\lg^4 n \lg\lg n}$ bits,
and they may be built for $\Oh{n/\lg^4 n}$ chunks, adding up to 
$\Oh{n\lg\lg n}$ bits. This is $\Oh{n\lg\sigma}$ if we assume $\lg\sigma=
\Omega(\lg\lg n)$.

\subsection{Queries} \label{sec:largesigma-queries}

Given $i$ and $j$, we compute $b = \lfloor\lg(j-i+1)\rfloor$ and
$t= \lceil \lg(1/\tau)\rceil$, and find the $\Oh{1}$ blocks of length $2^b$ 
overlapping $S[i..j]$. As in the previous section, every $G_b^t[k]=1$ in 
$G_b^t[i..j]$ is a candidate to verify, but this time we need to find its
leftmost occurrence in $S[i..j]$.

To find the leftmost position of $a = S[k]$, we see if the positions $k$ and 
$i$ are in the same chunk. That is, we compute the chunk index
$q=\lceil \rank_a(S,k)/\lg^4 n\rceil$ of $k$ (via a partial rank on $S$) and its
limits $i_l =\select_a(S,(q-1)\lg^4 n)$ and $i_r=\select_a(S,q\lg^4 n)$. 
Then we see if $i_l < i \le i_r$. In this case, we use the succinct SB-tree
associated with $J_b^t[i_r]=1$ to find the successor of $i$ in time $\Oh{1}$. 
Then we use Lemma~\ref{lem:check} from that position to determine in
$\Oh{1}$ time if $a$ is a $\tau$-majority in $S[i..j]$.

If $k$ is not in the same chunk of $i$, we disregard it because,
in this case, there is an occurrence $S[i_l]=a$ in $S[i..j]$ that is marked 
in $J_b^t$. We will instead find separately the leftmost occurrence in 
$S[i..j]$ of any candidate $a$ that is marked in $J_b^t[i..j]$, as follows.
We apply \citeauthor{Mut02}'s algorithm on the 1s of $J_b^t[i..j]$,
to find the distinct elements of $S_b^t[\rank_1(J_b^t,i-1)+1,\rank_1(J_b^t,j)]$.
Thus we obtain the leftmost sampled occurrences in $S[i..j]$ of all the 
$\tau$-majorities, among other candidates. For each leftmost occurrence 
$S_b^t[k']$, it must be that $k=\select_1(J_b^t,k')$ is in the same chunk of $i$, and 
therefore we can find the successor of $i$ using the corresponding succinct
SB-tree in constant time, and then
verify the candidate using Lemma~\ref{lem:check}.

It follows from the construction of $J_b^t$ that the distinct elements sampled
in any $S[i..j]$ must appear at least $2^t$ times in an interval of size
$\Oh{2^b}$ containing $S[i..j]$, and so there can only be $\Oh{1/\tau}$
distinct sampled elements. Therefore, \citeauthor{Mut02}'s algorithm on $J_b^t[i..j]$
gives us $\Oh{1/\tau}$ candidates to verify, in time $\Oh{1/\tau}$.

When $b < t$, we use our sequential algorithm of 
Section~\ref{sec:smallsigma-query} with the only difference that, since we
always find the leftmost occurrence of each candidate in $S[i..j]$, we can use
Lemma~\ref{lem:check} to verify the $\tau$-majorities. Thus the algorithm uses 
only select and partial rank queries on $S$, and therefore it runs in
time $\Oh{1/\tau}$ as well.

%OJO Better not to say this. Nobody said we couldn't report repetitions ;-) and
% it may give trouble for the mos compressed parts
%Note that we may find the same candidates a constant number of times, as they
%may be reported using $G_b^t$ by each of the up to 5 blocks overlapping 
%$S[i..j]$, and also using $J_b^t$. Thus we use a bitvector of size 
%$\sigma$ where we mark the majorities we find, to avoid repetitions. 
%If $\sigma > n$, we may instead use $j-i+1 \le n$ bits to mark the leftmost 
%positions tried in $S[i..j]$, which are a function of the symbol tested.

\begin{theorem} \label{thm:linear-largesigma}
Let $S[1..n]$ be a string over alphabet $[1..\sigma]$, with 
$\lg\sigma=\Omega(\lg\lg n)$. We can store $S$ in $\Oh{n\lg\sigma}$
bits such that later, given the endpoints of a range and $\tau$, we can return 
the $\tau$-majorities for that range in time $\Oh{1/\tau}$.
\end{theorem}

\subsection{Compressed space} \label{sec:largesigma-compressed}

To reduce the space, we use the same strategy used to prove
Theorem~\ref{thm:compressed-smallsigma}: we represent $S$ using the optimally
compressed structure of \citet{BCGNN13}. This time, however, closer to the
original article, we use different representations for the strings $S_c$ with
alphabets of size $\sigma_c \le w$ and of size $\sigma_c > w$. For the
former, we use the representation of Theorem~\ref{thm:succinct-smallsigma}, 
which uses $n_c\lg\sigma_c+o(n_c)$ bits and answers $\tau_c$-majority queries
in time $\Oh{1/\tau_c}$. For the larger alphabets, we use a slight variant
of Theorem~\ref{thm:linear-largesigma}: we use the same structures $G_b^t$, 
$J_b^t$, $C_b^t$, and pointers to succinct SB-trees, except that the lower 
bound for $b$ will be $\lfloor \lg (2^t \cdot g(n,\sigma)) \rfloor$, for any 
function $g(n,\sigma) = \omega(1)$. The total space for the 
bitvectors $G_b^t$ of string $S_c$ is thus
$\Oh{\frac{n_c\lg\sigma_c \lg g(n,\sigma)}{g(n,\sigma)}}=o(n_c\lg\sigma_c)$,
whereas the other structures already used $o(n_c)$ bits (with a couple of
exceptions we consider soon).

Then, representing $S_c$ with the structure of \citet[Thm.\ 6]{BN15}, so that
it supports
select in time $\Oh{g(n,\sigma)}$ and access in time $\Oh{1}$, the total space 
for $S_c$ is $n_c\lg \sigma_c + o(n_c\lg\sigma_c)$, and the whole structure
uses $nH + o(n)(H+1)$ bits. 

The cases where $b \ge \lfloor \lg (2^t \cdot g(n,\sigma)) \rfloor$ are solved
with $\Oh{1/\tau_c}$ applications of select on $S$, and therefore take time 
$\Oh{(1/\tau_c)\,g(n,\sigma)}$. Instead, the shorter ranges, of length
$\Oh{(1/\tau_c)\,g(n,\sigma)}$, must be processed sequentially, as
in Section~\ref{sec:smallsigma-query}. The space of the sequential algorithm 
can be maintained in $\Oh{1/\tau_c}=\Oh{1/\tau}$ words as follows. We cut the 
interval $S_c[i_c..j_c]$ into chunks of $m=\lceil 1/\tau_c\rceil$ consecutive 
elements, and process each chunk in turn as in 
Section~\ref{sec:smallsigma-query}. The difference is that we maintain an array
with the $\tau_c$-majorities $a$ we have reported and the last position $p_a$ we
have deleted in the lists. From the second chunk onwards, we remove all the 
positions of the known $\tau_c$-majorities $a$ before processing it, 
$\select_a(S_c,\rank_a(S_c,p_a)+r)$, for $r=1,2,\ldots$; note that 
$\rank_a(S_c,p_a)$ is a partial rank query. Since select on $S_c$ costs 
$\Oh{g(n,\sigma)}$ and we perform $\Oh{(1/\tau_c)\,g(n,\sigma)}$ operations, the
total time is $\Oh{(1/\tau_c)\,g(n,\sigma)^2}$. Then we can retain the optimally
compressed space and have any time of the form $\Oh{(1/\tau)\,f(n,\sigma)}$ by 
choosing $g(n,\sigma)=\sqrt{f(n,\sigma)}$.

There are, as anticipated, two final obstacles related to the space. The first
are the $\Oh{n_c}$ bits of \citeauthor{Mut02}'s structure associated with $S_c$
to handle the case $\tau_c < 1/\sigma_c$. To reduce this space to $o(n_c)$, we 
sparsify the structure as in Section~\ref{sec-min:compressed}. The case of 
small $\tau_c$ is then handled in time $\Oh{\sigma_c\,g(n,\sigma)^2}=
\Oh{(1/\tau_c)\,f(n,\sigma)}$ and the space for the sparsified structure is 
$\Oh{n_c/g(n,\sigma)}=o(n_c)$.

The second obstacle is the $\Oh{n_c \lg\lg n_c}$ bits used by the
succinct SB-trees. Examination of the proof of Lemma 3.3 in \citet{GRR09} 
reveals that
one can obtain $\Oh{p\lg\lg u}$ bits of space and $\Oh{\lg p / \lg\lg n}$ time
if we have $p$ elements in a universe $[1..u]$ and can store a precomputed
table of size $o(n)$ that is shared among all the succinct SB-trees. We reduce
the universe size as follows. We logically cut the string $S_c$ into 
$n_c/\sigma_c^2$ pieces of length $\sigma_c^2$. For each symbol $a$ we store a 
bitvector $B_a[1,n_c/\sigma_c^2]$ where $B_a[i]=1$ if and only if $a$ appears 
in the $i$th piece. These bitvectors require $\Oh{n_c/\sigma_c}$ bits in total, 
including support for rank and select.
The succinct SB-trees are now local to the pieces: a succinct SB-tree that spans
several pieces is split into several succinct SB-trees, one covering the
positions in each piece. The 1s corresponding to these pieces in bitvectors 
$B_a$ point to the newly created succinct SB-trees. To find the successor 
of position $i$ given that it is in the same chunk of $i_r > i$, with 
$J_b^t[i_r]=1$, we first compute the piece $p=\lceil i/\sigma_c^2\rceil$ of
$i$ and the piece $p_r=\lceil i_r/\sigma_c^2\rceil$ of $i_r$, and see if $i$
and $i_r$ are in the same piece, that is, if $p=p_r$. If 
so, the answer is to be found in the succinct SB-tree associated with the 1 at 
$J_b^t[i_r]$. Otherwise, that original structure has been split into several, 
and the part that covers the piece of $i$ is associated with the 1 at $B_a[p]$.
It is possible, however, that there are no elements in the piece $p$, that is,
$B_a[p]=0$, or that there are elements but no one is after $i$, that is, the
succinct SB-tree associated with piece $p$ finds no successor of $i$. In this 
case, we find the next piece that follows $p$ where $a$ has occurrences, 
$p'=\select_1(B_a,\rank_1(B_a,p)+1)$, and if $p' < p_r$ we query the succinct 
SB-tree associated with $B_a[p']=1$ for its first element (or the successor of 
$i$). If, instead, $p' \ge p_r$, we query instead the succinct SB-tree
associated with $J_b^t[i_r]=1$, as its positions are to the left of those 
associated with $B_a[p']=1$. Therefore,
successor queries still take $\Oh{1}$ time. The total number of elements stored
in succinct SB-trees is still at most $n_c$, because no duplicate elements are 
stored, but now each requires only $\Oh{\lg\lg \sigma_c}$ bits, for a total
space of $\Oh{n_c\lg\lg \sigma_c}=o(n_c\lg\sigma_c)$ bits. There may be up to 
$\sigma\cdot(n_c/\sigma^2)$ pointers to succinct SB-trees from bitvectors $B_a$,
each requiring $\Oh{\lg n_c}$ bits, for a total of $\Oh{(n_c\lg n_c)/\sigma_c} =
\Oh{n_c} = o(n_c\lg\sigma_c)$, since $\sigma_c > w$.

\begin{theorem} \label{thm:compressed-largesigma}
Let $S[1..n]$ be a string whose distribution of symbols has entropy $H$, over
alphabet $[1..\sigma]$. For any $f(n,\sigma)=\omega(1)$, we can store $S$ in 
$nH + o(n)(H+1)$ bits  such that 
later, given the endpoints of a range and $\tau$, we can return the 
$\tau$-majorities for that range in time $\Oh{(1/\tau)\,f(n,\sigma)}$.
\end{theorem}

Note that accessing a position in $S$ still requires constant time with this
representation. Further, we can obtain a version using nearly compressed space,
$(1+\epsilon)nH+o(n)$ bits for any constant $\epsilon>0$, with optimal query 
time, by setting $g(n,\sigma)$ to a constant value. First, use for $S_c$ the 
structure of \citet{BCGNN13} that needs $(1+\epsilon/3)nH+o(n)$ bits and solves
access and select in constant time. Second, let $\kappa$ be the constant
associated with the $\Oh{\frac{n_c\lg\sigma_c\lg g(n,\sigma)}{g(n,\sigma)}}$
bits used by bitvectors $G_b^t$ and the sparsified \citeauthor{Mut02}'s
structures. Then, choosing 
$g(n,\sigma)=\frac{6\kappa}{\epsilon} \lg\frac{6\kappa}{\epsilon}$ ensures 
that the space becomes $(\epsilon/3)n_c\lg\sigma_c$ bits, which add up to 
$(\epsilon/3)nH$. All the other terms of the form $o(nH)$ are 
smaller than another $(\epsilon/3)nH+o(n)$. 
Therefore the total space adds up to $(1+\epsilon)nH+o(n)$ bits.
The time to sequentially solve a range of length $\Oh{(1/\tau_c)\,g(n,\sigma)}$
is $\Oh{(1/\tau_c)\,g(n,\sigma)^2}=\Oh{1/\tau_c}$. 

\begin{theorem} \label{thm:epsilon-largesigma}
Let $S[1..n]$ be a string whose distribution of symbols has entropy $H$, over
alphabet $[1..\sigma]$. For any constant $\epsilon>0$, we can store $S$ in 
$(1+\epsilon)nH + o(n)$ bits  such that 
later, given the endpoints of a range and $\tau$, we can return the 
$\tau$-majorities for that range in time $\Oh{1/\tau}$.
\end{theorem}

%\section{Related Problems}
%
%A simpler variant of the problem we have considered arises when $\tau$ is fixed
%at preprocessing time. In this case we can obtain optimally compressed space and
%$\Oh{1/\tau}$ time, because we can use the construction of 
%Section~\ref{sec:smallsigma-struc} for one single value of $t=\lceil \lg(1/\tau)
%\rceil$. Therefore the bitvectors $G_b^t$ add up to $\Oh{n}$ bits. If 
%$\lg\sigma = \Oh{\lg w}$, we do not even need this improvement, as
%Theorem~\ref{thm:compressed-smallsigma} already gives us the desired result.
%For larger $\sigma$, we proceed as when proving 
%Theorem~\ref{thm:compact-largesigma}. The difference is that now the strings
%$S_c$ for $\sigma_c > w$ will require $n_c \lg \sigma_c + o(n\lg\sigma_c)
%+ \Oh{n} = n_c\lg\sigma_c +
%o(n_c \lg\sigma_c)$ bits, which will yield the optimally compressed space.
%
%\begin{theorem} \label{thm:fixedtau}
%Let $S[1..n]$ be a string whose distribution of symbols has entropy $H$, over
%alphabet $[1..\sigma]$, and a parameter $\tau$. We can store $S$ in 
%$nH + o(n)(H+1)$ bits such that later, given the endpoints of a range,
%we can return the $\tau$-majorities for that range in time $\Oh{1/\tau}$.
%\end{theorem}

\subsection{Finding range modes}

While finding range modes is a much harder problem in general,
we note that we can use our data structure from
Theorem~\ref{thm:compressed-largesigma} to find a range mode quickly when it
is actually reasonably frequent.  Suppose we want to find the mode of \(S
[i..j]\), where it occurs $\occ$ times (we do not know $\occ$).  We perform multiple range $\tau$-majority queries on \(S [i..j]\), starting with \(\tau = 1\) and repeatedly reducing it by a factor of 2 until we find at least one $\tau$-majority.  This takes time
\[\Oh{\left(1 + 2 + 4 + \ldots + 2^{\left\lceil \log \frac{j - i + 1}{\occ} \right\rceil} \right) f(n,\sigma)}
= \Oh{\frac{(j - i + 1) f(n,\sigma)}{\occ}}\]
and returns a list of $\Oh{\frac{j - i + 1}{\occ}}$ elements that
includes all those that occur at least $\occ$ times in $S[i..j]$.  We use rank
queries to determine which of these elements is the mode. For the fastest
possible time on those rank queries, we use for $S$ the representation of
\citet[Thm.\ 8]{BN15}, and also set $f(n,\sigma)=\lg\lg_w \sigma$, the same time of rank. The cost is then $\Oh{\frac{(j - i + 1) \log \log_w \sigma}{\occ}}$. The theorem holds for $\lg\sigma=\Oh{\lg w}$ too, as in this case we can use
Theorem~\ref{thm:compressed-smallsigma} with constant-time rank queries.

\begin{theorem} \label{thm:frequent modes}
Let $S[1..n]$ be a string whose distribution of symbols has entropy $H$. 
We can store $S$ in \(n H + o (n) (H + 1)\) bits
such that later, given endpoints $i$ and $j$, we can return the mode of \(S [i..j]\) in $\Oh{\frac{(j - i + 1) \log \log_w \sigma}{\occ}}$ time, where $\occ$ is the number of times the
mode occurs in $S[i..j]$.
\end{theorem}

\section{Range Minorities Revisited} \label{sec:minagain}

The results obtained for range majorities can be adapted to find range 
minorities, which in particular improves the result of
Theorem~\ref{thm-min:epsilon}. The main idea is again that, if we
test any $\lceil 1/\tau\rceil$ distinct elements, we must find a 
$\tau$-minority because not all of those can occur more than
$\tau(j-i+1)$ times in $S[i..j]$. Therefore, we can use mechanisms similar
to those we designed to find $\Oh{1/\tau}$ distinct candidates to 
$\tau$-majorities.

Let us first consider the bitvectors $G_b^t$ defined in 
Section~\ref{sec:smallsigma}. We now define bitvectors $I_b^t$, where 
we flag the positions of the first $2^t$ and the last $2^t$ distinct values 
in each block (we may flag fewer positions if the block contains less than
$2^t$ distinct values). Since we set $\Oh{2^t}$ bits per block, the bitvectors 
$I_b^t$ use asymptotically the same space of the bitvectors $G_b^t$.

Given a $\tau$-minority query, we compute $b$ and $t$ as in 
Section~\ref{sec:smallsigma} and use rank and select to find all
the 1s in the range $I_b^t[i..j]$. Those positions contain a $\tau$-minority 
in $S[i..j]$ if there is one, as shown next.

\begin{lemma}
The positions flagged in $I_b^t[i..j]$ contain a $\tau$-minority in $S[i..j]$,
if there is one.
\end{lemma}
\begin{proof}
If $I_b^t[i..j]$ overlaps a block where it does not flag 
$2^t=\lceil 1/\tau\rceil$ distinct elements (in which case one is for sure
a $\tau$-minority), then it marks all the distinct block elements that fall 
inside $[i..j]$. This is obvious if $[i..j]$ fully contains the block, and
it also holds if $[i..j]$ intersects a prefix or a suffix of the block, since
the block marks its $2^t$ first and last occurrences of distinct elements.
\end{proof}

Just as for $\tau$-majorities, we use $I_b^t$ only if $1/\tau \le \sigma$, 
since otherwise we can test all the alphabet elements one by one. The test
proceeds using rank on $S$ if $\sigma$ is small, or using Lemma~\ref{lem:check}
if $\sigma$ is large. We now describe precisely how we proceed.

\subsection{Small alphabets} \label{sec:min-smallsigma}

If $\lg\sigma = \Oh{\lg w}$, we use a multiary wavelet tree as in 
Section~\ref{sec:smallsigma}. This time, we do not run $\tau$-majority queries
on each wavelet tree node $v$ to determine which of its children to explore, but
rather we explore every child having some symbol in the range $S_v[i_v..j_v]$. 
To efficiently find the distinct symbols that appear the range, we store a 
sparsified \citeauthor{Mut02}'s structure similar to the one described in 
Section~\ref{sec-min:compressed}; this time we will have no slowdown
thanks to the small alphabet of $S_v$. 

Let $C_v$ be the array corresponding to 
string $S_v$. We cut $S_v$ into blocks of $w^\beta$ bits, and record in an array
$C'_v[1..n_v/w^\beta]$ the minimum value in the corresponding block of $C_v$.
Then, the leftmost occurrence $S[k]=p$ of each distinct symbol $p$ in 
$S_v[i_v..j_v]$ has a value $C_v[k]<i_v$, and thus its corresponding block 
$C_v'[k']$ also holds $C_v'[k']<i_v$. We initialize a word $E \leftarrow 0$ 
containing flags for the $\sigma'=w^\beta$ symbols, separated as in the final 
state of the word $C$ of Section~\ref{sec:smallsigma-succinct}. Each time the 
algorithm of \citeauthor{Mut02} on $C_v'$ gives us a new block, we apply the
algorithm of Section~\ref{sec:smallsigma-succinct} to count in a word $C$ the 
occurrences of the distinct symbols in that block, we isolate the counters
reaching the threshold $y=1$, and compare $E$ with $E~\textsc{or}~C$. If they
are equal, then we stop the recursive algorithm, since all the symbols in the
range had already appeared before (see the final comments on the proof of
Lemma~\ref{lem:sparsemuthu}). Otherwise, we process the subrange to the left
of the block, update $E \leftarrow E~\textsc{or}~C$, and process the subrange
to the right. When we finish, $E$ contains all the symbols that appear in
$S_v[i_v..j_v]$. In the recursive process, we also stop when we have considered
$\lceil 1/\tau_v \rceil$ blocks, since each includes at least one new element
and it is sufficient to explore $\lceil 1/\tau_v \rceil$ children to find a 
$\tau_v$-minority (because each child contains at least
one candidate). Finally, we extract the bits of $E$ one by one as done in 
Section~\ref{sec:smallsigma-succinct} with the use of $D$. For each extracted
bit, we enter the corresponding child in the wavelet tree. The total time
is thus $\Oh{1/\tau_v}$ and the bitvectors $C_v$ add up to 
$\Oh{n/w^\beta}=o(n)$ bits in total.

The $\tau_u$ values to use in the children $u$ of $v$ are computed
as in Section~\ref{sec:smallsigma-succinct}, so the analysis
leading to $\Oh{1/\tau}$ total time applies. When we arrive at the leaves $u$ of
the wavelet tree, we obtain the distinct elements and compute using rank 
the number of times they occur in $S_u[i_u..j_u]$, so we can immediately 
report the first $\tau$-minority we find.

We still have to describe how we handle the intervals that are smaller than
the lower limit for $b$, $\lfloor 2^t \cdot w^\beta/4 \rfloor$. We do the
counting exactly as in Section~\ref{sec:smallsigma-succinct}. We must then 
obtain the counters that are between 1 and $y-1$. On one hand, we use the 
bound $y'=1$ and repeat their computation to obtain in $C_l \leftarrow C$ 
the counters that are at least 1. On the other, we compute
$C \leftarrow C + (2^{2\ell}-y) \cdot (0^{k\ell+\ell-1} 1 0^{(k-1)\ell})^{\sigma'}$
as before, and isolate the non-overflowed bits with
$C_r \leftarrow (\textsc{not}~C)~\textsc{and}~(0^{(k-1)\ell-1} 1 0^{(k+1)\ell})^{\sigma'}$.
Then we extract the first of the bits marked in $C \leftarrow C_l~\textsc{and}~
C_r$ and report it.

To obtain compressed space, we use the alphabet partitioning technique of
Section~\ref{sec:smallsigma-compressed}. Once again, we must identify at most
$\lceil 1/\tau \rceil$ nonempty ranges $[i_c..j_c]$ from $K[i..j]$. Those are
obtained in the same way as on the multiary wavelet tree, since $K$ is 
represented in that way (albeit the strings $S_v$ are compressed).
We then look for $\tau_c$-minorities in the strings 
$S_c[i_c..j_c]$ one by one, until we find one or we exhaust them. The total
time is $\Oh{1/\tau}$.

\begin{theorem} \label{thm-min:compressed-smallsigma}
Let $S[1..n]$ be a string whose distribution of symbols has entropy $H$, over
alphabet $[1..\sigma]$, with $\lg\sigma=\Oh{\lg w}$. We can store $S$ in
\(n H + o(n)\) bits such that later, given the endpoints of a range and $\tau$,
we can return a $\tau$-minority for that range (if one exists) in time 
$\Oh{1/\tau}$.
\end{theorem}

Note that we can use a single representation using $nH+o(n)$ bits solving both
the $\tau$-majority queries of Theorem~\ref{thm:compressed-smallsigma} and
the $\tau$-minority queries of Theorem~\ref{thm-min:compressed-smallsigma}.

\subsection{Large alphabets}

For large alphabets we must use Lemma~\ref{lem:check} to check for 
$\tau$-minorities, and thus we must find the leftmost positions in
$S[i..j]$ of the $\tau$-minority candidates. We use the same bitvectors
$J_b^t$ of Section~\ref{sec:largesigma}, so that they store sampled positions 
corresponding to the 1s in $I_b^t$, and proceed exactly as in that section,
both if $\tau < 1/\sigma$ or if $\tau \ge 1/\sigma$.

To obtain compression, we also use alphabet partitioning. We use on the
multiary wavelet tree of $K$ the method described in
Section~\ref{sec:min-smallsigma}, and then complete the queries with
$\tau_c$-minority queries on the strings $S_c$ over small or large alphabets,
as required, until we find one result or exhaust all the strings.
The only novelty is that we must now find $\tau_c$-minorities sequentially for 
the ranges that are shorter than $\lfloor \lg(2^t \cdot g(n,\sigma)) \rfloor
= \Oh{(1/\tau_c)\,g(n,\sigma)}$. For this, we adapt the 
$\Oh{(1/\tau)\,g(n,\sigma)^2}$-time sequential algorithm 
described in Section~\ref{sec:largesigma-compressed}. The only difference is 
that we stop as soon as we test a candidate $a$ that turns out not to be a 
$\tau_c$-majority, then reporting the $\tau$-minority $a$.

Depending on whether we use 
Theorem~\ref{thm:compressed-largesigma} or \ref{thm:epsilon-largesigma} to
represent $S$ and how we choose $f(n,\sigma)$, we obtain 
Theorem~\ref{thm-min:compressed} again or an improved version of
Theorem~\ref{thm-min:epsilon}:

%\begin{theorem} \label{thm-min:compressed-largesigma}
%Let $S[1..n]$ be a string whose distribution of symbols has entropy $H$, over
%alphabet $[1..\sigma]$. For any $f(n,\sigma)=\omega(1)$, we can store $S$ in 
%$nH + o(n)(H+1)$ bits  such that later, given the endpoints of a range and 
%$\tau$, we can return a $\tau$-minority for that range (if one exists) in time 
%$\Oh{(1/\tau)\,f(n,\sigma)}$.
%\end{theorem}

\begin{theorem} \label{thm-min:epsilon-largesigma}
Let $S[1..n]$ be a string whose distribution of symbols has entropy $H$, over
alphabet $[1..\sigma]$. For any constant $\epsilon>0$, we can store $S$ in 
$(1+\epsilon)nH + o(n)$ bits  such that later, given the endpoints of a range 
and $\tau$, we can return a $\tau$-minority for that range (if one exists) in 
time $\Oh{1/\tau}$.
\end{theorem}

In both cases, we can share the same structures to find majorities and 
minorities.

\section{Conclusions} \label{sec:conclusions}

We have given the first linear-space data structure for parameterized range majority with query time $\Oh{1 / \tau}$, even in the more difficult case of $\tau$ specified at query time. This is worst-case optimal in terms of $n$ and $\tau$, since the output size may be up to $1/\tau$.  Moreover, we have improved the space bounds for parameterized range majority and minority, reaching in many cases optimally compressed space with respect to the entropy $H$ of the distribution of the symbols in the sequence. While we have almost closed the problem in these terms, there are some loose ends that require further research:
\begin{itemize}
\item Our results for $\tau$-majorities are worst-case time optimal, but they 
take $\Oh{1/\tau}$ time even if the number of majorities is $o(1/\tau)$. Is it
possible to run in time $\Oh{\occ+1}$ when there are $\occ$ $\tau$-majorities?
Can we use $\Oh{\occ+1}$ instead of $\Oh{1/\tau}$ space?
\item Our structure and previous ones for $\tau$-minorities also take time
$\Oh{1/\tau}$, although we are required to output only one 
$\tau$-minority. Is it possible to run in $\Oh{1}$ time, or to prove a lower
bound? Can we use less than $\Oh{1/\tau}$ space?
\item On large alphabets, $\sigma=w^{\omega(1)}$, both for $\tau$-majorities 
and $\tau$-minorities we must use $(1+\epsilon)nH+o(n)$ bits, for any 
constant $\epsilon>0$, to reach time $\Oh{1/\tau}$. With $nH+o(n)(H+1)$ bits 
we only have a time of the form $(1/\tau)\cdot \omega(1)$. Is it possible to
close this gap?
\item Our results do not improve when $\tau$ is fixed at indexing time, which
is in principle an easier scenario. Is it possible to obtain better results 
for fixed $\tau$?
\end{itemize}

\begin{acks}
Many thanks to Patrick Nicholson for helpful comments.
\end{acks}

\bibliographystyle{authordate1}
\bibliography{paper}

\end{document}